\documentclass[10pt, conference, letterpaper]{IEEEtran}

\usepackage{amsmath}
\usepackage{amsthm}
\usepackage{graphicx}
\usepackage{color}
\usepackage{amssymb}
\usepackage[draft]{fixme}
\usepackage{soul}
\usepackage{epsfig}
\usepackage{mathptmx}
\usepackage{times}
\usepackage{amssymb}
\usepackage{mathtools}
\usepackage{geometry}

\newtheorem{theorem}{\textbf{Theorem}}
\newtheorem{definition}{\textbf{Definition}}
\newtheorem{remark}{\textbf{Remark}}
\newtheorem{assumption}{\textbf{Assumption}}
\newtheorem{lemma}{\textbf{Lemma}}

\begin{document}
\newgeometry{top=1in,bottom=0.75in,outer=0.75in,inner=0.75in}
\title{\LARGE \bf
Resilient Energy Allocation Model for Supply Shortage Outages}

\author{
    \IEEEauthorblockN{Miguel Alberto C. Mercado}
    \IEEEauthorblockA  {
        University of the Philippines, Diliman\\
        mcmercado4@up.edu.ph
    }
    \and
    \IEEEauthorblockN{Roy Dong}
    \IEEEauthorblockA  {
        University of California, Berkeley\\
        roydong@eecs.berkeley.edu
    }
    \and
    \IEEEauthorblockN{Allan Nerves}
    \IEEEauthorblockA  {
        University of the Philippines, Diliman\\
        anerves@eee.upd.edu.ph
    }
}

\maketitle
\thispagestyle{empty}
\pagestyle{empty}

\begin{abstract}
Supply Shortage Outages are a major concern during peak demand for developing countries. In the Philippines, commercial loads have unused backup generation of up to 3000 MW, at the same time there are shortages of as much as 700 MW during peak demand. This gives utilities the incentive to implement Demand Response programs to minimize this shortage. But when considering Demand Response from a modeling perspective, social welfare through profit is always the major objective for program implementation. That isn't always the case during an emergency situation as there can be a trade-off between grid resilience and cost of electricity.

The question is how the Distribution Utility (DU) shall optimally allocate the unused generation to meet the shortage when this trade-off exists. We formulate a combined multi-objective optimal dispatch model where we can make a direct comparison between the least-cost and resilience objectives. 

We find that this trade-off is due to the monotonically increasing nature of energy cost functions. If the supply is larger than the demand, the DU can perform a least-cost approach in the optimal dispatch since maximizing the energy generated in this case can lead to multiple solutions. We also find in our simulation that in cases where the supply of energy from the customers is less than shortage quantity, the DU must prioritize maximizing the generated energy rather than minimizing cost.
\end{abstract}
\section{INTRODUCTION}
There is potential to increase the reliability\footnote{Or resiliency.} of electricity grids during Supply Shortage Outages (SSO). This can done through the use of the in-house distributed generation of commercial and industrial loads. In fact, in the Philippines, there is as much as 3000 MW worth of unused distributed power through these generators even with energy crisis of as much as 700MW~\cite{intro}. Therefore, there is an incentive for the regulator to implement a Demand Side Management (DSM) program~\cite{DSM} to allow distribution utilities (DU) to contract certain commercial establishments to either (1) reduce energy consumption, or (2) use their in-house generators during peak demand. We put our attention on the latter case.

In the literature, managing energy demand through incentive contracts fall in the realm of Demand Response (DR)~\cite{DR-berk}. The US Department of Energy defines Demand Response as "Changes in electric usage by end-use customers from their normal consumption patterns in response to changes in the price of electricity over time, or to incentive payments designed to induce lower electricity use at times of high wholesale market prices or when system reliability is jeopardized"~\cite{DR_def}. Specifically, we are dealing with the Interruptible Load Program (ILP)~\cite{DR-berk}.

For most cases in DR, and DSM in general, profits, social welfare, and prices are put into the spotlight in the design of these programs~\cite{DSM, DR-berk}. In certain cases, such as emergency situations, this approach may leave the implementation of certain DR schemes ineffective in increasing resiliency. Most of the time during emergency situations, the amount of shortage in a distribution system would be too high for the PCs to be able to fully compensate the shortage. This is because generation is usually not part of the core competencies of the PCs.

Resilience must then be explicitly considered because there may be a trade-off between the minimization of cost, and the minimization of the shortage. Therefore, this paper focuses on creating a mechanism that considers both the allocation of the amount of energy each PC must generate through a least-cost optimal dispatch, or by increasing grid resiliency. This is done through a modified optimal energy dispatch model for the distributed generators of the PCs~\cite{elec-market}.

The rest of the paper is organized as follows: in Section II, we review the different DR mechanism with the goal of optimizing social welfare through profit~\cite{MWG}. In Section III, we develop the agent model and make assumptions on their cost functions. In Section IV and V, we formulate the optimal dispatch models, and analyze their properties. In Section VI, we show a simulation on the proposed model. We conclude this paper in Section VII.
\section{LITERATURE REVIEW}
DR techniques are used to shift consumption patterns as a solution to reduce energy shortage~\cite{DSM, DR-berk}. The drawback to using DR is that these methods do not necessarily reduce energy consumption as shifting demand may lead to profiles with new peaks~\cite{DSM}. To fix this problem, DR methods can also involve methods that induce consumers to generate their own energy~\cite{DR-berk}.

There are two kinds of DR strategies that are currently being used: incentive-based and time-based DR~\cite{DR-berk}. Incentive-based DR schemes refer to programs where DUs give incentives to customers to influence consumption and generation, while time-based DR schemes refer to mechanisms where DUs modify prices throughout the day to influence consumption patterns~\cite{DR-over}.

Work has been done to discuss a competitive market model for demand response assuming firms are price-takers under a profit-maximization model~\cite{RRL1, RRL2}. Gadham and Ghose show how social welfare, as the sum of consumer and producer surplus, is affected when the equilibrium price change through a sensitivity analysis~\cite{RRL3}. Chen and Low model social welfare maximization through a residential demand response scheme where each household operates different appliances~\cite{RRL4}. Nijhuis, et al. implemented a demand response program at a franchise level using the elasticity of  aggregated demand~\cite{RRL5}. Also, we see that Sabounchi, et al. defined social welfare as the preferences of a consumer using the aggregate energy consumed of all devices in the household~\cite{RRL6}. Chen, et al. created two market models for DR under the assumption of a perfect competition, and an oligopoly~\cite{RRL7}. Also, Su, and Kirschen quantify how social welfare, defined as the sum of the utility functions, is affected as the number of participating customers (PC) increases~\cite{RRL8}.

Common among the incentive-based schemes is that DUs contract to affect consumption rather than explicitly contract to generate energy. Another issue with most DR mechanisms is that they mostly consider social welfare through profit optimization. In cases where there are SSOs, this objective may lead to the ineffective implementation of the program because of a possible trade-off between cost and resilience.

With the Philippines, since there are a lot commercial establishments have their own distributed generators, it is more efficient to contract with the latter case. With that in mind, it is also not explicitly mentioned how the DU optimally allocates the amount of shortage each PC is required to generate. Therefore, we focus our attention on how the DU shall allocate the amount of shortage through a modified optimal dispatch~\cite{elec-market}.

Traditionally, optimal dispatch problems have the objective of minimizing the cost to buy a certain amount of power from the generation sector~\cite{elec-market}. Barley and Winn study choosing an optimal energy mix with corresponding dispatch strategies for microgrids~\cite{RRL9}. Traditionally, optimal dispatch problems are implemented for generators competing in the transmission system~\cite{RRL10}. Also, it is standard to model dispatch problems as a deterministic model, but Viviani and Heydt model stochastic optimal dispatch problems where the parameters are random variables, and the agents are risk-averse~\cite{RRL11}. Research has also been done by Baran and El-Markabi to model optimal dispatch in distribution systems with the objective of minimizing the total reactive power dispatched to regulate the voltage in distribution feeders~\cite{RRL12}. Also, DallAnese, et al. implemented the optimal dispatch of photovoltaic systems in residential distribution systems~\cite{RRL13}.

Common among all the optimal dispatch models is the knowledge that the market is operating with a business as usual. Since we are dealing with SSOs, the objective of the DU in the dispatch may change drastically. Therefore in our case, we will study whether it is still efficient to implement a least-cost optimal dispatch when resilience can also be a big issue in cases when there are SSOs.
\section{AGENT MODEL}
\subsection{Agent Cost Function}
In this section, we discuss the preliminary definitions and assumptions for the model. We begin by defining the costs for each PC.
\begin{definition}
	For a set of PCs \(\{1,...,N\}\) where the PC \(i \in \{1,...,N\}\), Let \(\theta_i \in \Theta_i\) denote the unit price of generating \(\bar{e}_i \in E_i \subseteq \mathbb{R}_{\geq 0}\) amount of energy for the PC.
	
	We define the cost of the PC \(i\) to generate \(\bar{e}_i\) to be \(C_i: E_i \times \Theta_i \rightarrow \mathbb{R}_{\geq 0}\).
\end{definition}
In a standard optimal dispatch problem, this is the cost to produce \(\bar{e}_i\) amount of energy for PC \(i\). The goal of the DU is to minimize the overall cost to society as mandated by the system regulator.

We state the assumed properties of the cost function model for each PC to generate his own electricity. We base this on the standard cost model~\cite{MWG} used in microeconomics.

\begin{assumption}
	\label{cost-assumption}
	The generator cost function \(C_i(\bar{e}_i, \theta_i) \in C^2\) is continuous, twice differentiable\footnote{\(C^2\) is the space of continuous functions with 2nd derivatives}, and monotonically increasing in \(\bar{e}_i\) and \(\theta_i\), while convex in \(\bar{e}_i\) and concave in \(\theta_i\).
\end{assumption}

This assumption is reasonable as most distributed generators do have a convex cost function~\cite{elec-market}.

Lastly, we make define the DU demand.

\begin{definition}
	The total DU demand is \(\bar{e} \in \mathbb{R}_{\geq 0}\).
\end{definition}

This implies that the DU demand \(\bar{e}\) is an external variable that is dependent on the nature of the SSO. In the economics literature, this means that the \(\bar{e}\) is exogenous~\cite{MWG}.

\subsection{Agent Generator Constraints}
In this section, we characterize the capacity of the generators being operated by the PC. In a market setting, it is sufficient to characterize the generator using an inequality constraint~\cite{elec-market}.
\begin{assumption}
	The PC has a maximum average power consumption of \(P^{max}_i \in \mathbb{R}_{> 0}\).
\end{assumption}
This can be justified because PCs generally have finite average power consumption. Since most ILPs last only a few hours, the amount of energy they would need to generate is finite.

Next we need to assume the minimum power consumption each generator shall take. In the literature~\cite{elec-market}, it is usually assumed that each generator generates a minimum of \(0W\). In our case, generating \(0W\) is akin to not participating in the contract at all, and is undesirable when one signs up for ILP.
\begin{assumption}
	The PC is always required to join the program and generate energy with a minimum average power consumption of \( P^{min}_i > 0\) where \(P^{min}_i \in \mathbb{R}\) and \(P^{max}_i > P^{min}_i\).
\end{assumption}
Since ILP programs mostly have energies as an input instead of power~\cite{ILP}, the next remark combines both the previous assumptions.
\begin{remark}
	The PC satisfies the inequality constraint \(0 < P^{min}_i \leq \frac{\bar{e}_i}{T} \leq P^{max}_i\) for the optimal dispatch problem.
\end{remark}
In this model, \(\frac{\bar{e}_i}{T} = P_i\) where \(P_i\) is assumed to be the average power consumed during the ILP.
\section{Cost-Based Optimal Dispatch Model}
In this Section, we discuss the least-cost based model for being able to do optimal energy dispatch.
\subsection{Market Clearing Constraint}
First, we state an assumption\footnote{Only valid in this section.} on the nature of the supply of the DGs.
\begin{assumption}
	\label{market-clearing-assumption}
	The PCs can always supply as much as the DU demands.
\end{assumption}
The previous assumption is only valid when the sum of the maximum generator constraints satisfy the following relationship:
\begin{equation}
	\label{market-clearing}
	T\sum_{i=1}^{N}P_i^{max} \geq \bar{e}.
\end{equation}

This is a strong assumption because it is not always the case that the DGs of the PCs can supply enough to clear the shortage. We will see in the section V that if this property does not hold, the feasibility set in the optimization problem is empty.

With this, the market clearing constraint~\cite{MWG} is
\begin{equation}
	\label{market-clear}
	\sum_{i=1}^N\bar{e}_i = \bar{e}.
\end{equation}

The market clearing constraint implies that any combination of \(\bar{e}_i\) \(\forall i\) can be used as long as it satisfies \eqref{market-clear}. We call this homogeneity since each \(\bar{e}_i\) can be replaced by any \(\bar{e}_j\) for \(j \neq i\)~\cite{MWG}.
\subsection{Cost Optimization Problem}
We base our optimization problem on a dispatch that minimizes the total cost to generate \(\bar{e}\). Therefore, with the market clearing constraint and the generation constraints, the optimization problem is
\begin{equation}
\label{standard_op_prob}
	\min_{[\bar{e}_1, \bar{e}_2, ..., \bar{e}_N]} \sum_{i=1}^{N}C_i(\bar{e}_i, \theta_i).
\end{equation}
subject to:
\begin{displaymath}
	P^{min}_i \leq \frac{\bar{e}_i}{T} \leq P^{max}_i \text{ } \forall i
\end{displaymath}
\begin{displaymath}
	\sum_{i=1}^N\bar{e}_i = \bar{e}.
\end{displaymath}

The next Lemma shows when the problem is feasible with a global optimal solution~\cite{elec-market}.
\begin{lemma}
	\label{lemma1}
	If \(T\sum_{i=1}^{N}P_i^{max} \geq \bar{e}\) and Assumption \ref{cost-assumption} is satisfied, then there exists a global optimal solution, \([\bar{e}_1, ..., \bar{e}_N]^T\), that satisfies \eqref{standard_op_prob}.
\end{lemma}
Intuitively, this means that the sum of the maximum capacities must be greater than the demand for there to be a solution. If the premise is not satisfied, then no combination of \([\bar{e}_1, ..., \bar{e}_N]^T\) will satisfy the market clearing constraint. We treat this case in the next section.
\section{Modified Optimal Dispatch Model for Emergency Situations}
In reality, such as in emergency situations due to SSOs, the amount of shortage is almost always too high such that the total amount of energy the PCs generate is not enough to cover the shortage. With that, In this section, we analyze this case where Lemma \ref{lemma1} because Assumption \ref{market-clearing-assumption} is not satisfied.  More concisely, we consider the case where the DGs don't have the capacity to clear the shortages.

To preempt the results from this section, we show that the trade-off between resilience and cost is due to the inefficiency of the least-cost objective in minimizing the shortage. This is due to Assumption \ref{cost-assumption}. With that, we form a multi-objective optimal dispatch problem that can characterize this trade-off.

\subsection{Market Clearing Constraint Breakdown}
The next Lemma formally states what happens when Assumption \ref{market-clearing-assumption}, which assures Lemma \ref{lemma1}, does not hold.
\begin{lemma}
	\label{lemma2}
	If \(T\sum_{i=1}^{N}P_i^{max} < \bar{e}\), then the feasible set that satisfies \eqref{standard_op_prob} is empty.
\end{lemma}
This just means that if the individual maximum capacities don't add up to the shortage, then there is no way the market clears. This is an issue because Lemma \ref{lemma2} implies there is no feasible solution to the optimal dispatch problem.

We combine Lemma \ref{lemma1} and Lemma \ref{lemma2} to produce the following Theorem.
\begin{theorem}
	\label{theorem1}
		\(T\sum_{i=1}^{N}P_i^{max} \geq \bar{e}\) iff \(\exists\) a \([\bar{e}_1, ..., \bar{e}_N]^T\) that satisfies \eqref{standard_op_prob} and Assumption \ref{cost-assumption}.
\end{theorem}
Theorem \ref{theorem1} provides a necessary and sufficient condition for the feasibility of \eqref{standard_op_prob}. 

When considering emergency situations where supply of energy from the PCs can't meet demand, as formalized in Theorem~\ref{theorem1}, we quantify this insight by relaxing the market cleating constraint with
\begin{equation}
	\label{ineq-market-clearing}
	\sum_{i=1}^N\bar{e}_i \leq \bar{e}.
\end{equation}

In this case where supply can't meet demand, the aggregate energy supply the PCs produce will always be less than \(\bar{e}\).
\subsection{Least-Cost Objective Trade-Off}
As stated previously, we are dealing with the problem of minimizing the total cost of the \(N\) PCs subject to the relaxed market clearing constraint.

An issue here is that minimizing the cost that is monotonically increasing subject to \eqref{ineq-market-clearing} leads to a trade-off with grid resiliency - a big objective of DR. 
\begin{theorem}
	\label{theorem2}
	If \(\exists\) a \([\bar{e}_1, ..., \bar{e}_N]^T\) that satisfies \eqref{standard_op_prob} while relaxing the market clearing constraint, then \(\bar{e}_i = TP_i^{min}\).
\end{theorem}
This is because with a monotonically increasing cost function found in Assumption \ref{cost-assumption}, \([\bar{e}_1, ..., \bar{e}_N]^T\) would minimize the total cost by dispatching the minimum amount of energy, \(TP_i^{min}\), the PCs can generate while in the program.

This result is troublesome as the ultimate goal of a DR scheme during SSOs is to minimize the total shortage. Using a least-cost objective would necessarily decrease the system's resilience, and would render any incentive-based DR scheme for energy generation during emergency situations useless. Therefore, we consider adding to the least-cost objective the objective of minimizing the shortage in the sequel.

\subsection{Resilience Objective}
Since the DU cannot tell PCs to generate sufficient energy, we add a second objective to maximize the amount of energy generated. This is equivalent to minimizing the shortage during the SSO. Therefore, we consider another objective
\begin{equation}
	\label{new-obj}
	\max_{[\bar{e}_1,...,\bar{e}_i]}\sum_{i=1}^N\bar{e}_i.
\end{equation}

We characterize the solution to this problem in the next Remark.

\begin{lemma}
	\label{lemma3}
	The optimization of \eqref{new-obj} with the same constraints as \eqref{standard_op_prob} has multiple optimizers if the generator constraints for each PC \(i\) are not all tight at the optimal solution of the new problem.
\end{lemma}

Therefore, for cases where supply of energy is greater than demand, the DU still has to decide which optimal solution it will use for the allocation. In this case, we recommend the DU to use the standard least-cost dispatch.
\subsection{Multi-Objective Optimal Dispatch}
Lastly, because of the trade-off between least cost and least shortage, DUs generally don't want to focus a big chunk their resources in minimizing their energy due to a loss of revenue. Therefore, we formulate the modified optimal dispatch problem as a multi-objective convex optimization problem~\cite{Convex-Optimization}.

 With the objectives and the constraints laid out, the generalized optimal dispatch problem, after regularization~\cite{Convex-Optimization}, is then
\begin{equation}
\label{relaxed-optimization}
\min_{[\bar{e}_1, ..., \bar{e}_i]}\lambda \sum_{i=1}^{N}C_i(\bar{e}_i, \theta_i) -(1-\lambda)\sum_{i=1}^N\bar{e}_i.
\end{equation}
subject to:
\begin{displaymath}
P^{min}_i \leq \frac{\bar{e}_i}{T} \leq P^{max}_i \text{ } \forall i
\end{displaymath}
\begin{displaymath}
\sum_{i=1}^N\bar{e}_i \leq \bar{e}.
\end{displaymath}
Since the model is convex~\cite{Convex-Optimization}, all we need to do is prove that there is a feasible set to state a necessary and sufficient condition for global optimality.
\begin{theorem}
	\label{theorem3}
	The feasible set of \eqref{relaxed-optimization} is non-empty.
\end{theorem}

Notice how the regularization of the problem differs slightly from the standard Tikhonov Regularization~\cite{Convex-Optimization}. This has some analytic and interpretive significance. The solution to the optimal dispatch problem is going to be a Pareto Frontier~\cite{Convex-Optimization} from \(\lambda \in [0,1]\) instead of \(\lambda \in \mathbb{R}_{\geq 0}\). 

The value of \(\lambda\) then could be interpreted as the ratio of how much the DU should prioritize ones cost objective over the resilience objective. If \(\lambda =0\) then the firm prioritizes resilience. The cost objective is prioritized when \(\lambda = 1\). With this framework, firms can do a sensitivity analysis on \(\lambda\) to see whether or not they should opt to prioritize minimizing the shortage or minimizing their costs.  If we characterize the Pareto Frontier of the optimal solution, we are able to see the effectiveness of the program when one prioritizes one cost over resilience.

\section{APPLICATION AND SIMULATION RESULTS}
In this section, we provide a simple analysis of the claim above. Consider the quadratic fuel cost function \eqref{fuel-cost}. We consider three cases for this application:
\begin{enumerate}
	\item \(\lambda\) Sensitivity Analysis
	\item Cost-Based Optimal Dispatch
\end{enumerate}
\subsection{Assumed Cost Function}
In general, any cost function that satisfies the properties found in any standard graduate microeconomic theory textbook~\cite{MWG} will work, but as is usual for power engineering, we modify the standard quadratic fuel cost curve for the purpose of application~\cite{elec-market}.

Assuming the generator generates an average power \(\frac{\bar{e}}{T}\), the fuel cost function is 

\begin{equation}
\label{fuel-cost}
C_i(\bar{e}_i, \theta_i)=T[a_2\theta_i(\frac{\bar{e}_i}{T})^2 + a_1\theta_i(\frac{\bar{e}_i}{T}) + a_0\theta_i]
\end{equation}
for some cost parameters \(a_2, a_1, a_0 \in \mathbb{R}, a_2 > 0\), and length of shortage \(T\).

The quadratic component corresponds to all nonlinear variable costs present due to system nonlinearities such as starting and stopping the generator. The linear term is a variable cost that represents the cost of operating and maintaining the generator at steady-state. The \(a_0\) term represents all fixed costs~\cite{elec-market}.
\subsection{Optimal Dispatch Simulation Parameters}
We consider implementing the analysis of the optimal dispatch problems with 5 PCs. We  set \(P_i^{min}=30 kW\) \(\forall i\). The maximum average generation for each PC to be \(P^{max}=[60,100, 125, 85, 130]^T kW\). We consider an ILP where \(T=1 hr\). We consider the generator parameters given in Table \ref{gen-par}~\cite{gen-parameters}.
\begin{table} [ht]
	\caption{Generator Parameters}
	\centering
	\begin{tabular}{c c c c c c}
		\hline \hline
		Parameter & PC 1 & PC 2 & PC 3 & PC 4 & PC 5\\ [0.5ex]
		\hline
		\(a_0 \theta\)  & \(96.6\) & \(96.6046\) & \(96.279\) & \(100.3937\) & \(95.856\) \\
		\(a_1 \theta\) & \(7.588\) & \(7.5874\) & \(7.592\) & \(6.9761\) & \(7.374\) \\
		\(a_2 \theta\) & \(0.0414\) & \(0.0414\) & \(0.042\) & \(0.0533\) & \(0.047\) \\ [1ex]
		\hline
	\end{tabular}
	\label{gen-par}
\end{table}

\subsection{\(\lambda\) Sensitivity Analysis}
We consider \(P_i^{min}=30kW\) \(\forall i\), and \(\bar{e}=700kWhr\) because \(700 kWhr > 500kWhr = T\sum P_i^{max}\) In this case, we consider how the dispatch changes when the parameter \(\lambda\) changes.

Figure 1 shows how the dispatch is sensitive to \(\lambda\).
\begin{figure} [ht]
	\label{dispatch_700}
	\includegraphics[width=\linewidth]{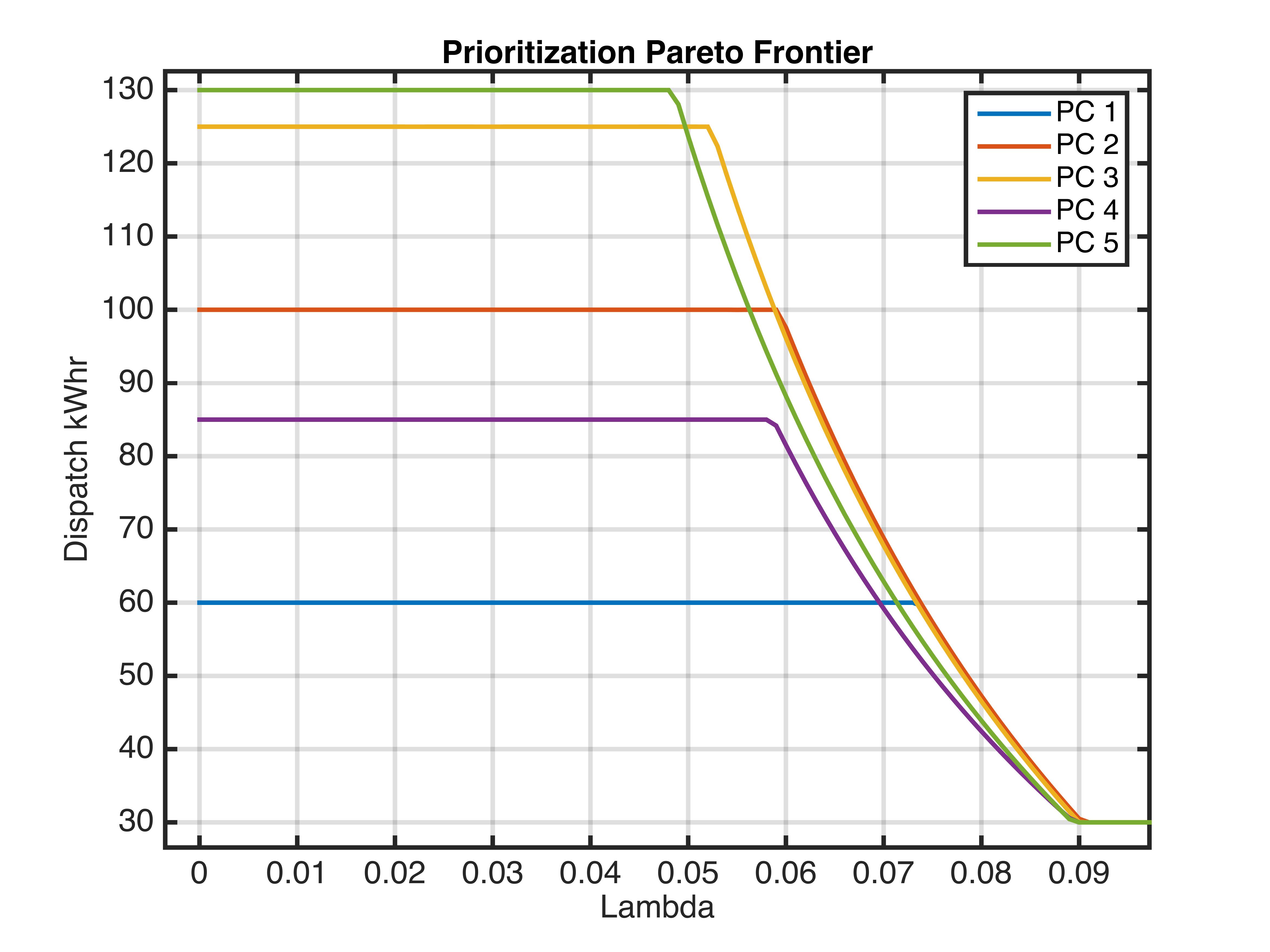}
	\caption{PC Relaxed Optimal Dispatch Sensitivity With Respect to \(\lambda\)}
\end{figure}

We observe that the dispatch remains constant until around \(\lambda=0.05\). This implies that the effect of minimizing costs starts to dominate even when the DU prioritizes maximizing energy generated more. In fact, the cost objective totally dominates when \(\lambda \geq 0.09\), which is less than \(\lambda=0.5\) where the DU prioritizes the two objectives equally. Therefore, when an SSO occurs during an emergency situation, the DU must decide to prioritize maximizing energy generated more than minimizing the total cost to generate.
\subsection{Cost-Based Optimal Dispatch}
For cases where the supply of electricity is greater than the demand due to a shortage, we use Remark \ref{lemma3} to justify using a Cost-Based Optimal Dispatch.

We consider checking the sensitivity of the Strict Dispatch from \(\bar{e}=150kWhr\) to \(\bar{e}=500kWhr\). \(150kWhr\) was chosen because the \(\sum P_i^{min} = 150kW.\)

Figure 2 shows how the optimal dispatch changes when the total shortage increases from \(150kWhr\) to \(500kWhr\).
\begin{figure} [ht]
	\label{sensitivity_DP}
	\includegraphics[width=\linewidth]{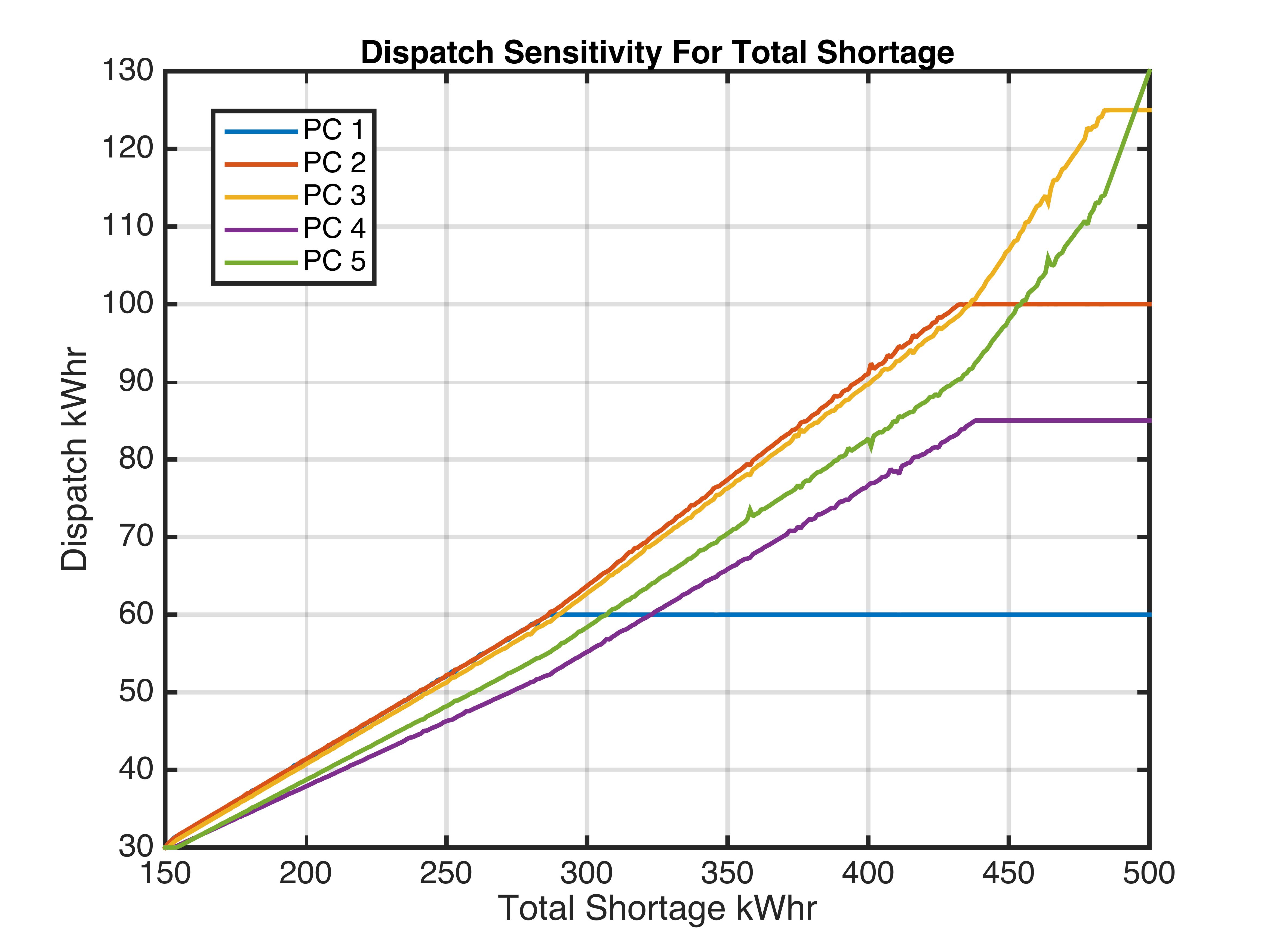}
	\caption{PC Dispatch Sensitivity With Respect to \(\bar{e}\)}
\end{figure}

It can be seen that the dispatch monotonically increases with respect to \(\bar{e}\) until one PC's full capacity is reached. We observe the behavior that when firm \(j\) reaches its full capacity, the rate of change of the dispatch \(\forall i \neq j\) increases. This implies that the DU demands more from the other PCs marginally when PC \(i\) reaches full capacity.

\section{CONCLUSION AND FUTURE WORK}
In this paper, we develop an optimal dispatch model for allocating generation during supply shortage outages (SSO) in emergency situations. This is for the development of the Interruptible Load Program. We develop a model for optimal dispatch for SSOs by modifying the least-cost optimal dispatch found in standard literature~\cite{elec-market}.

Using the modified model, we have found that when the supply of energy provided by the contracted establishments is less than the shortage demand, there is a trade-off between the resilience of the dispatch with the cost. This is because using the least-cost approach may lead to the inefficient implementation of the program since grid resilience will increase minimally. Mathematically, this inefficiency is caused by the reasonable assumption that costs are monotonically increasing in the amount of energy. Characterizing this inefficiency is crucial for the interruptible load program as grid resilience should be of top priority during emergency situations.

Though, for special cases where all the total number of PCs can cover the energy shortage, we recommend using a least-cost approach as it will not lead to the inefficient implementation of the mechanism.

In this model, we have considered PCs that are taken to be price-takers. Also, in this model, we have assumed that the generation prices and valuations are known such that the PCs can't cheat the system by raising prices. Lastly, the model breaks down in the case where the cost is not a linear function of price~\cite{MWG}.

In future studies, we consider energy with significant market power such that they can influence the price. We also consider the case where the energy valuations of each PC is unknown. This can be done by implementing through either a uniform or discriminatory price auction. Finally, we consider models where the costs are not a linear function of the unit price commonly found in contract models~\cite{MWG}.



\section*{APPENDIX}
In this section, we provide the proofs for Lemma \ref{lemma1}, Lemma \ref{lemma2}, Theorem \ref{theorem2}, Lemma \ref{lemma3}, and Theorem \ref{theorem3}. Theorem \ref{theorem1} is not included as it is a direct consequence of Lemmas \ref{lemma1} and \ref{lemma2}.
\begin{proof}[Proof of Lemma \ref{lemma1}]
	Let \(X_i =[P_i^{min}, P_i^{max}] \in \mathbb{R}_{> 0}\). Since \(P_i^{max} = \sup\{X_i\}\), \(\sum_{i=1}^{N}TP_i^{max} \geq \bar{e}\), and \(P_i^{max} \geq 0\), take \(\bar{P}_i \geq 0\) \(\forall i\) s.t. \(\sum_{i=1}^{N}T(P_i^{max} - \bar{P}_i) = \bar{e}\). 
	
	Since \(P_i^{max} - \bar{P}_i \in X_i\), then \(\exists\) \([x_1 \in X_1, x_2 \in X_2,...,x_N \in X_N]^T\) that satisfies \eqref{standard_op_prob}.
	
	Additionally, since Assumption \ref{cost-assumption} implies that the problem is convex~\cite{Convex-Optimization}, then the solution is a global optima.
\end{proof}
\begin{proof}[Proof of Lemma \ref{lemma2}]
	Let \(X_i =[P_i^{min}, P_i^{max}] \subset \mathbb{R}_{> 0}\). Since \(\sum_{i=1}^{N}TP_i^{max} < \bar{e}\), and \(P_i^{max} = \sup\{X_i\}\), then any \(x_i \in X_i\) \(\forall i\) will lead to \(\sum_{i=1}^{N}Tx_i < \sum_{i=1}^{N}TP_i^{max} < \bar{e}\).
	
	Therefore, \(\nexists\) \([x_1 \in X_1, x_2 \in X_2,...,x_N \in X_N]^T\) that satisfies \eqref{market-clearing}. Therefore, the feasible set of \eqref{standard_op_prob} is empty.
\end{proof}
\begin{proof}[Proof of Theorem \ref{theorem2}]
	Let \(X_i =[P_i^{min}, P_i^{max}] \in \mathbb{R}_{> 0}\). Since the cost function is monotonically increasing, \(\arg\min_{\bar{e}} C_i(\bar{e}_i, \theta_i)\) subject to the minimimum generation constraint is \(TP_i^{min}.\)
\end{proof}
\begin{proof}[Proof of Lemma \ref{lemma3}]
	Restating the contrapositive, we have "If the solution to \eqref{standard_op_prob} replacing the objective with \eqref{new-obj} is unique, then the generator constraints for each PC \(i\) are all tight in the new problem." We prove this by contradiction.
	
	Let \(\mathbf{e} = [\bar{e}_{11}, ..., \bar{e}_{1i}]\) be an optimal solution to the problem. Let \(x \in \mathbb{R}_{\geq 0}\) where \(\min\{\mathbf{e}\} < x < \max\{\mathbf{e}\}\) s.t. \(\min\{\mathbf{e}\} + x\) is less than \(TP^{max}_j\) for its corresponding 
	\(j \in \{1,...,N\}\). This is possible by the completeness of the reals.
	
	We can take the vector \(\mathbf{e}^* = [\bar{e}_{11}, ..., \bar{e}_{1i}]\) s.t. \(\max\{\mathbf{e}^*\} = \max\{\mathbf{e}\} - x\), and \(\min\{\mathbf{e}^*\} = \min\{\mathbf{e}\} + x\). Since this is also an optimal solution, we have a contradiction on uniqueness - proving the contrapositive. \(\perp\)
\end{proof}
\begin{proof}[Proof of Theorem \ref{theorem3}]
	Let \(X_i =[P_i^{min}, P_i^{max}] \in \mathbb{R}_{> 0}\). Take \(\frac{\bar{e}_i}{T}= \frac{P_i^{max}+P_i^{min}}{2}\). Thus, \(\sum_{i=1}^{N}\frac{\bar{e}_i}{T}=\sum_{i=1}^{N}\frac{P_i^{max}+P_i^{min}}{2} \leq \sum_{i=1}^{N}P_i^{max}\). 
	
	Therefore, since \(\sum_{i=1}^{N}\frac{\bar{e}_i}{T}\leq \sum_{i=1}^{N}P_i^{max}\), then \(\sum_{i=1}^{N}\bar{e}_i\leq \sum_{i=1}^{N}TP_i^{max}=\bar{e}\). Therefore, The feasible set is non-empty.
\end{proof}
\section*{ACKNOWLEDGMENT}
This work was done through the financial support from the Philippine Commission of Higher Education under the Philippine-California Advanced Research Institution program.


\end{document}